\documentclass{llncs}
\usepackage{amsmath}
\usepackage{amssymb}
\usepackage{graphicx}
\usepackage{algorithm}
\usepackage{algorithmicx}
\usepackage{algpseudocode}

\def\c{\ensuremath{\mathbf{c}}}
\def\cost{\mathit{cost}}
\def\OPT{\mathit{OPT}}
\def\APP{\mathit{APP}}
\def\ALG{\mathit{ALG}}

\newcommand{\remove}[1]{}

\newenvironment{varalgorithm}[1]
  {\algorithm\renewcommand{\thealgorithm}{#1}}
  {\endalgorithm}

\title{On the Tree Search Problem with Non-uniform Costs}

\author{
Ferdinando Cicalese\inst{1}
\and
Bal\'azs Keszegh\inst{2}
\thanks{Research supported by Hungarian National Science Fund (OTKA), under
grant PD 108406 and under grant NN 102029 (EUROGIGA project GraDR
10-EuroGIGA-OP-003) and the J\'anos Bolyai Research Scholarship of the
Hungarian Academy of Sciences. }
\and
Bernard Lidick\'{y}\inst{3}
\thanks{Research is partially supported by NSF grant DMS-1266016.}
\and
D\"om\"ot\"or P\'alv\"olgyi\inst{4}
\thanks{Research supported by Hungarian
National Science Fund (OTKA), under grant PD 104386 and under grant NN 102029
(EUROGIGA project GraDR 10-EuroGIGA-OP-003) and the J\'anos Bolyai Research
Scholarship of the Hungarian Academy of Sciences.}
\and
Tom\'a\v{s} Valla\inst{5}
\thanks{Supported by the Centre of Excellence -- Inst.\ for Theor.\ Comp.\ Sci.
(project P202/12/G061 of GA~\v{C}R).}
}
\institute{Department of Computer Science, University of Salerno, Italy
\and
R\'enyi Institute, Hungary
\and
University of Illinois at Urbana-Champaign, USA, {\tt lidicky@illinois.edu}
\and
E\"otv\"os University, Hungary
\and
Czech Technical University, Faculty of Information Technology, Prague, Czech Republic,
{\tt tomas.valla@fit.cvut.cz}
}

\date{}

\begin{document}

\maketitle

\begin{abstract}
Searching in partially ordered structures has been considered in the context of
information retrieval and efficient tree-like indexes, as well as in hierarchy
based knowledge representation.
In this paper we focus on tree-like partial orders and consider the problem of
identifying an initially unknown vertex in a tree by asking edge queries:
an edge query $e$ returns the component of $T-e$ containing the vertex sought for,
while incurring some known cost $c(e)$.

The Tree Search Problem with Non-Uniform Cost is:
given a tree $T$ where each edge has an associated cost,
construct a strategy that minimizes the total cost of the identification in the worst case.

Finding the strategy guaranteeing the minimum possible cost is an NP-complete
problem already for input tree of degree 3 or diameter 6.
The best known approximation guarantee is the $O(\log n/\log \log \log n)$-approximation
algorithm of [Cicalese et al. TCS 2012].

We improve upon the above results both from the algorithmic and the
computational complexity point of view:
We provide a novel algorithm that provides an $O(\frac{\log n}{\log \log n})$-approximation
of the cost of the optimal strategy.
In addition, we show that finding an optimal strategy is
NP-complete even when the input tree is a spider, i.e., at most one vertex has
degree larger than 2.
\end{abstract}

\section{Introduction}
The design of efficient procedures for searching in a discrete structure is a
fundamental problem in discrete mathematics \cite{Ahlswede,Aigner} and computer science \cite{Knuth}.
Searching is a basic primitive for building and managing operations of an information system as ordering,
updating, and retrieval.
The typical example of a search procedure is binary search which allows to
retrieve an element in a sorted list of size $n$ by only looking at $O(\log n)$
elements of the list.
If no order can be assumed on the list, then it is known that any procedure
will have to look at the complete list in the worst case.
Besides these two well characterized extremes, extensive work has also been
devoted to the case where  the underlying structure of the search space is a partial order.
Partial orders can be used to model lack of information on the totally ordered
elements of the search space \cite{Linial} or can naturally arise from the
relationship among the elements of the search space,
like in hierarchies used to model knowledge representation \cite{Wermelinger},
or in tree-like indices for information retrieval of large databases \cite{Ben-Asher}.
For more about applications of tree search see below.

In this paper, we focus on the case where the underlying search space is a
tree-like partially ordered set and tests have nonuniform costs.
We investigate the following problem.

\begin{figure}
\centering
\includegraphics{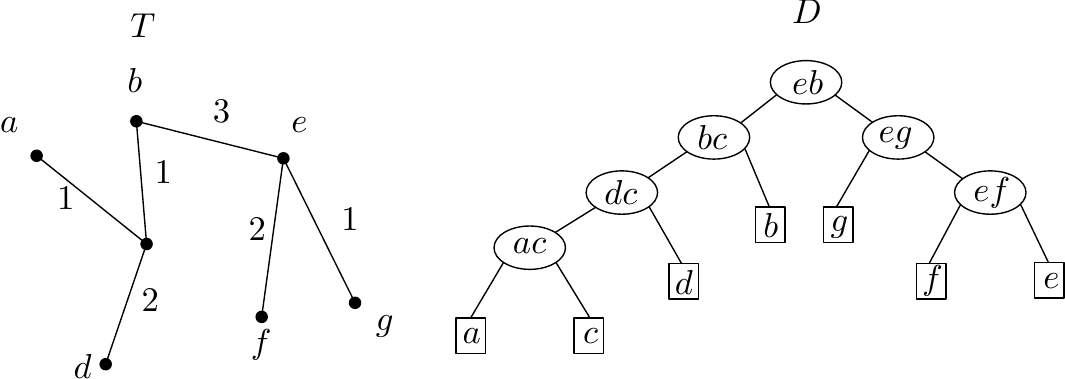}
\caption{An example of the tree search problem,
$T$ is the input tree and $D$ is a decision tree with $cost(D) = 7 = cost^D(a) = cost^D(c)$.
If the vertices of the tree $T$ represent the parts of a device to assemble,
the decision tree corresponds to the assembly procedure that at time 0 joins $e$ with $b$;
then at time $3$ joins $b$ with $c$ and $e$ with $g$.
At time $4$ the joining of  $d$ with $c$ and $e$ with $f$ is started.
Finally, at time $6$ part $a$ is joined with part $c$ and the procedure ends by time $7$.}
\label{fig:problem-defi}
\end{figure}

\noindent
{\sc The Tree Search Problem with non-uniform costs}

{\em Input}: A tree $T=(V, E)$ with non-negative rational costs assigned to the edges
defined by a $c : e \in E \mapsto c(e) \in \mathbb{Q}$.

{\em Output}: A strategy that minimizes (in the worst case) the cost spent to identify
an initially unknown vertex $x$ of $T$ by using {\em edge queries}.
An {\em edge query} $e = \{u, v\} \in E$ asks for the subtree $T_u$ or $T_v$
which contains $x$, where $T_u$ and $T_v$ are the
connected components of $T-e$, including the vertex $u$ and $v$ respectively.
The cost of the query $e$ is $c(e)$.
The cost of identifying a vertex $x$ is the sum of the costs of the queries asked.

\smallskip

More formally, a strategy for the Tree Search Problem with nonuniform costs
over the tree $T$ is a  \emph{decision tree} $D$ which is a rooted binary tree
with $|V|$ leaves where every leaf $\ell$
is associated with one vertex $v \in V$ and every internal node\footnote{For
the sake of avoiding confusion between the input tree and the decision tree, we
will reserve the term vertex for the elements of $V$ and the term $node$ for
the vertices of the decision tree $D$.}
$\nu \in V(D)$ is associated with one test $e = \{u, v\} \in E$.
The outgoing edges from  $\nu$ are associated with the possible outcomes of the
query, namely, to the case where the vertex to identify lies in $T_u$ or $T_v$ respectively.
Every vertex has at least one associated leaf.
The actual identification process can be obtained  from $D$ starting with the
query associated to the root and moving towards the leaves based on the answers received.
When a leaf $\ell$ is reached, the associated vertex is output
(see Fig.~\ref{fig:problem-defi} for an example).

Given a decision tree $D$, for each vertex $v\in V(T)$,
let $\cost^D(v)$ be the sum of costs of the edges associated
to nodes on the path from the root of $D$ to the leaf identifying $v$.
This is the total cost of the queries performed when the
strategy $D$ is used and $v$ is  the vertex to be identified.

In addition, let the cost of $D$ be defined by
$$
\cost(D) = \max_{v\in V(T)} \cost^D(v).
$$
This is the worst-case cost of identifying a vertex of $T$ by the decision tree $D$.
The optimal cost of a decision tree for the instance represented by the tree $T$
and the cost assignment $\c$ is given by
$$
\OPT(T,\c) = \min_{D} \cost(D),
$$
where the $\min$ is over all decision trees $D$ for the instance $(T, c)$.

\smallskip
\noindent
{\bf Previous results and related work.}
The Tree Search Problem has been first studied under the name of
tree edge ranking \cite{IyeRatVij91,laTGreSch95,Lam98optimaledge,Makino2001,Dereniowski08},
motivated by multi-part product assembly.
In \cite{Lam98optimaledge} it was shown that in the case where the tests have uniform cost,
an optimal strategy can be found in linear time.
A linear algorithm for searching in a tree with uniform cost was also provided in \cite{Mozes}.
Independently of the above articles, the first paper where the problem is
considered in terms of searching in a tree is \cite{Ben-Asher}, where the more
general problem of searching in a poset was also addressed.

The variant considered here in which the costs of the tests are non-uniform
was first studied by Dereniowski \cite{Dereniowski06} in the context of edge ranking.
In this paper, the problem was proved NP-complete for trees of diameter at most 10.
Dereniowski also provided  an $O(\log n)$ approximation algorithm.
In \cite{CicaleseJLV12} Cicalese et al.\  showed that the tree search problem
with non-uniform costs is strongly NP-complete already for input trees of
diameter 6, or maximum degree 3, moreover, these results are tight.
In fact, in \cite{CicaleseJLV12}, a polynomial time algorithm computing the
optimal solution is also provided for diameter 5 instances and an $O(n^2)$
algorithm for the case where the input tree is a path.
For arbitrary trees, Cicalese et al.\ provided
an $O(\frac{\log\ n}{\log \log \log n})$-approximation algorithm.


%
%
%
%

\smallskip
\noindent
{\bf Our Result.} Our contribution is both on the algorithmic and on the complexity side.
On the one hand, we provide a new approximation algorithm for the tree search problem with non-uniform costs
which improves upon the best known guarantee given in \cite{CicaleseJLV12}.
In Section 3 we will prove the following result.

\begin{theorem}\label{thm:main}
There is an $O(\log n / \log\log n)$-approximation algorithm
for the Weighted Tree Search Problem that runs in polynomial time in $n$.
\end{theorem}

In addition, we show that the tree search problem with non-uniform costs is NP-hard already
when the input tree is a spider\footnote{By \emph{spider} we mean a tree
with at most one vertex of degree greater than 2.}
of diameter 6.

\smallskip

\noindent
{\bf More about applications.}
We discuss some scenarios in which
the problem of searching in trees with non-uniform costs naturally arises.

Consider the problem of locating a buggy module in a program in which the
dependencies between different modules can be represented by a tree.
For each module we can verify the correct behavior independently.
Such a verification may consist in checking, for instance, whether all branches
and statements in a given module work properly. For different modules, the cost
of using the checking procedure can be different (here the cost might refer to
the time to complete the check). In such a situation, it is important to device
a debugging strategy that minimizes the cost incurred in order to locate the
buggy module in the worst case.

Checking for consistency in different sites keeping distributed copies of
tree-like data structures (e.g., file systems) can be performed
by maintaining at each node some check sum information about the subtree rooted at that node.
Tree search can be used to identify the presence of ``buggy nodes'', and
efficiently identifying the  inconsistent part in the structure, rather than
retransmitting or exhaustively checking  the whole data structure.
In \cite{Ben-Asher}, an application of this model in the area of information retrieval is also described.

Another examples comes from a class of problems which is in some sense dual to the previous ones:
deciding the assembly schedule of a multi-part device.
Assume that the set of pairs of parts that must be assembled together can be represented by a tree.
Each assembly operation requires some (given) amount of time to be performed
and while assembling two pieces, the same pieces cannot  be involved in any other assembly operation.
At any time different pairs of parts can be assembled in parallel. 
The problem is to define the schedule of assembly operations
which minimize the total time spent to completely assembly the device. 
The schedule is an edge ranking of the tree defined by the assembly operations. 
By reversing the order of the assembly operation in the schedule we obtain a
decision tree for the problem of searching in the tree of assembly operation
where each edge cost is equal to the cost of the corresponding assembly.

\section{Basic lower and upper bounds}

In this section we provide some preliminary results which will be useful in the
analysis of our algorithm presented in the next section. We introduce some
lower bounds on the cost of the optimal decision tree for a given instance
of the problem. We also recall  two exact algorithm for constructing optimal
decision tree which were given in  \cite{CicaleseJLV12}. The first is an
exponential time dynamic programming algorithm which works for any input tree.
The second is a  quadratic time algorithm for instances
where the input tree is a path.
Finally, we show a construction of 2-approximation decision trees for spider graphs.

\smallskip

Let $T$ denote the input tree and $\c$ the cost function.
It is not hard to see that, given a decision tree $D$ for $T$ we can extract
from it a decision tree for the instance of the problem
defined on a subtree $T'$ of $T$ and the restriction of $\c$ to the vertices in $T'$.
For this, we can repeatedly apply the following operation: if in $D$ there is a
node $\nu$ associated with an edge $e = \{u, v\}$, such that $T_u$ (reps.
$T_v$) is included $T - T'$ then remove the node $\nu$ together with the
subtree rooted at the child of $\nu$ corresponding to the case
where the vertex to identify is in $T_u$ (resp.\ $T_v$).
Let $D'$ be the resulting decision tree when the above step cannot be performed anymore.
Then, clearly $cost(D', \c) \leq cost(D, \c)$. We have shown the following (also observed in \cite{CicaleseJLV12}).

%
%
%

\begin{lemma}\label{lem:subtree}
Let $T'$ be a subtree of $T$. Then, $\OPT(T, \c) \geq \OPT(T', \c)$.
\end{lemma}


%

Another immediate observation is that for a given input tree $T$, the value
$OPT(T, \c)$ is monotonically non-decreasing with respect to the cost
of any edge. This is recorded in the following.
\begin{lemma}\label{lem:decreasec}
Let $\c$ and $\c'$ be cost assignments on a tree $T$ such that
$\c'(e) \leq \c(e)$ for every $e \in E(T)$.
Then, $\OPT(T,\c) \geq \OPT(T,\c')$.
\end{lemma}
%

The next proposition shows that subdividing an edge cannot decrease the cost of the optimal decision tree.

\begin{proposition}\label{prop:contract}
Let $\c$ be a cost assignment on a tree $T$.
Let $v\in V(T)$ have exactly two neighbors $u_1,u_2 \in V(T)$.
If $T'$ is obtained from $T-v$ by adding the edge $\{u_1,u_2\}$ and
$\c'$ is obtained from $\c$ by setting $\c'(u_1u_2) = \min\{\c(u_1v),\c(u_2v)\}$
then $\OPT(T,\c) \geq \OPT(T',\c')$.
\end{proposition}

The proof of Proposition~\ref{prop:contract} is deferred to the appendix.


The following two results  from \cite{CicaleseJLV12} provide exact algorithms for the construction of
optimal strategies.
More precisely, Proposition \ref{prop:exp}
provides  an exponential dynamic programming based algorithm for general trees.
Theorem \ref{thm:path} gives an $O(n^2)$ time algorithm for
the special case where the input tree is a path
and will be useful in the analysis of our main algorithm and also in the following lemma
regarding the spider tree.

\begin{proposition}[\cite{CicaleseJLV12}]\label{prop:exp}
Let $T$ be an edge-weighted tree on $n$ vertices.
Then an optimal decision tree for $T$ can be constructed in $O(2^n n)$ time.
\end{proposition}



The following theorem was proved by Cicalese et al.\ in \cite{CicaleseJLV12}
and will be useful later in the analysis of our algorithm and also in the following lemma
regarding the spider tree.

\begin{theorem}[\cite{CicaleseJLV12}]\label{thm:path}
There is an $O(n^2)$ time algorithm that constructs an optimal decision tree $D$ for
a given weighted path on $n$ vertices.
\end{theorem}

Note that for a star $T$ any decision tree $D$ has the same cost, since all the edges have to be asked in the worst case.
Hence, for a tree $T$ such that there is only one node with degree greater than $1$
we have $OPT(T, \c) = \sum_{e \in E(T)} c(e),$ for any cost function $\c$.

\begin{definition}
A tree $T$ is a \emph{spider} if there is at most one vertex in $T$ of degree greater than two.
We refer to this vertex as the {\em head} (or {\em center}) of the spider.
Moreover,  each path from the head of the spider to one of the leaves will be referred to as  a {\em leg} of the spider.
\end{definition}

\begin{lemma}\label{lem:spider}
Let $T$ be a spider.
Then there is an algorithm which computes a 2-approximate decision tree $D$ for $T$
and runs in time $O(n^2)$.
\end{lemma}

\begin{proof}
If $T$ is a path, then by Theorem~\ref{thm:path} there exists an algorithm computing
the optimal decision tree in $O(n^2)$ time.
Assume $T$ is not a path. Then $T$ contains exactly one vertex $v$ of degree at least three.
Let $S_v$ be the star induced by $v$ and the vertices adjacent to $v$. Let us denote by $w_1,\ldots,w_k$
 the vertices adjacent to $v$, where $k=\deg(v)$.
By Theorem~\ref{thm:path}, for every $i \in  \{1,\dots,k\}$ we construct the optimal decision tree $D_i$
for the path component $C_i$ of $T-v$ containing $w_i$ in time $O(|C_i|^2)$.
Note that the total running time for construction of $D_1,\dots,D_k$ is $O(n^2)$.
Finally, for $S_v$ we compute the optimal decision tree $D_v$ (in $O(n)$ time).
The decision tree $D$ for $T$ is obtained from $D_v$
by replacing the node corresponding to $w_i$ by the root of $D_i$ for every $i \in \{1,\dots,k\}$.
Clearly, the algorithm runs in $O(n^2)$ time
and $cost(D) \leq \OPT(S_v, \c)  + \max_{1\leq i\leq k}\{\OPT(C_i, \c)\}\leq 2\OPT(T, \c)$.
The last inequality follows since  by Lemma~\ref{lem:subtree} both $\OPT(S_v, \c)$ and
$\max_{1\leq i\leq k}\{\OPT(C_i, \c)\}$ are  lower bounds on  $OPT(T, \c)$.
\qed\end{proof}

\section{The Algorithm}

Let $n$ be the size of the input tree  and $t = 2^{\lfloor\log\log n\rfloor+2}$
be a parameter fixed for the whole run of the algorithm.
It holds that  $2 \log n\le t\le 4 \log n$.

The basic idea of our algorithm is to construct a subtree $S$ of the input tree $T$ such that:
(i) we can construct a decision tree for $S$ whose cost is
at most a constant times the cost of an optimal decision tree for $S$;
(ii) each component of $T -S$ has size not larger than $T/t$.

This will allow us to build a decision tree for $T$ by assembling the decision tree for $S$
with the decision trees recursively constructed for the components of $T-S$.
The constant approximation guarantee on $S$ and the fact that, due to the size of the subtrees on which we recurs,
we need at most $O(\frac{\log n}{\log \log n})$ levels of recursion to show
that our algorithm gives an $O(\frac{\log n}{\log \log n})$ approximation.

\smallskip

\noindent
{\bf The subtree $S$.}
We iteratively build subtrees $S_0 \subset S_1 \subset \cdots \subset S_t \subseteq T$.
Starting with the empty tree $S_0$,
in every iteration $i \in \{1,\dots,t\}$ we pick a centroid $x_i$ of the largest connected component of the forest $T-S_{i-1}$.
The subtree $S_i$ is set to be the minimal subtree containing $x_i$ and $S_{i-1}$.
If for some $i$ we have that $S_i=T$, then we set $S = S_i = T$ and we stop the iterations.
If all $t$ iterations are completed then we set $S=S_t$.

By definition, the \emph{centroid} of a tree $T$ is a vertex $v$ such that any maximal component of $T- v$
has size at most $|T|/2$.
Therefore, we have the following lemma---which  establishes (ii) above.

\begin{lemma} \label{lem:H-components}
If $H$ is a maximal connected component of $T - S$, then $|H|\le |T|/\log n$.
\end{lemma}
\begin{proof}

We prove by induction on $k$ that after $2^k$ iterations all maximal components of $T-S_{2^k}$ have size at most $|T|/2^{k-1}$.
Let $k = 0$. We observe that by the definition of centroid, after $1 = 2^{k}$ 
iterations all components of $T-S_1$ have size at most $|T|/2 \leq |T|/ 2^{k-1} = 2|T|$.
This establishes the basis of our induction. 

Now fix some $k > 0$ and assume (induction hypothesis) that after $2^{k-1}$ iterations all maximal components 
of $T-S_{2^{k-1}}$ have size at most $|T|/2^{k-2}$.
Among these there are at most $2^{k-1}$ components that have size at least $|T|/2^{k-1}$.
In the next $2^{k-1}$ iterations we will choose a centroid in each of these components, one by one.
Choosing a centroid in a component $H$ splits $H$ into parts that have size at most half of $H$,
thus after $2^k=2^{k-1}+2^{k-1}$ steps all components of $T-S_{2^k}$ have size at most $|T|/2^{k-1}$.

Thus, if the process of constructing $S$ is stopped after $t=2^{\lfloor\log\log n\rfloor+2}$ iterations
all components have size at most $|T|/2^{\lfloor\log\log n\rfloor+1}\le |T|/\log n$.
On the other hand, if the process of constructing $S$ is stopped at some iteration $i < t$ then it means that $S =T$ and,
trivially, we have $|H| = 0$.
\qed\end{proof}

\smallskip
\noindent
{\bf The Decision Tree for $S$.}
Let $X$ contain all $x_i$ for $i\in \{1,\dots,t\}$ and vertices of degree at least three in $S$.
Note that $|X|\le 2t+1$. 
Let $P_{u,v}$ be the path of $T$ whose endpoints are vertices $u$ and $v$.

We define an auxiliary tree $Y$ on the vertex set $X$.
Vertices $u,v \in X$ form an edge of $Y$
if $u$ and $v$ are the only vertices of $X$ of the path $P_{u,v}$ in $T$ with endpoints $u$ and $v$.
Let $e_{uv} = \arg \min_{e\in P_{u,v}} \c(e)$ (the edge of $P_{u,v}$ with minimal cost) and
$\c_Y(uv) = c(e_{uv})$.
Let $Z=\bigcup_{uv\in E(Y)} e_{uv}$. By Proposition~\ref{prop:exp},
we can compute an optimal decision tree $D_Y$ for $Y$  in $O(2^{2t}t)$ which is polynomial in $n$.

Let $D_X$ be obtained from $D_Y$ by changing the label of every internal node
from $uv$ to $e_{uv}$, for each $uv \in E(Y)$.
The tree $D_X$ is not a decision tree for $S$, however,
leaves of $D_X$ correspond to connected components of $S-Z$.
Notice that $cost(D_X) = cost(D_Y) = OPT(Y, c_Y)$.

Since every connected component $C$ of $S-Z$ contains at most one vertex of degree at least three,
every such component is a spider. By Lemma~\ref{lem:spider}, a decision tree $D_C$ for each such component $C \in S-Z$
can be computed in $O(n^2)$ time with approximation ratio 2.

We can now obtain the decision tree $D_S$ for $S$ by replacing each leaf in $D_X$ with the decision tree
for the corresponding component in $S-Z$. We have

\begin{eqnarray*} \label{eq:S-approx}
\frac{cost(D_S)}{OPT(S,c)} & \leq &\frac{cost(D_X) + \max_{C \in S-Z} cost(D_C)}{OPT(S, c)} \\
 & \leq &\frac{cost(D_X)}{OPT(Y, c_Y)} + \max_{C \in S-Z} \frac{cost(D_C)}{OPT(C, c)} \leq 3,
\end{eqnarray*}
where the second inequality holds because of  $OPT(Y, c_Y) \leq OPT(S, c)$ (given by Proposition \ref{prop:contract}) and
$OPT(C, c) \leq OPT(S, c)$ (given by Lemma \ref{lem:subtree}).

\smallskip
\noindent
{\bf Assembling the pieces in the Decision Tree for $T$}.
Let  $v$ be a vertex in $S$ with a neighbor not in $S$, let $S_v$ be the star induced by $v$ and its neighbors outside $V(S)$.

Let $D_v$ be a decision tree for $S_v$ (notice that they all have the same cost).
For every neighbor $w \not\in V(S)$ of $v$ we compute recursively the decision tree $D_w$ for
the component $H_w$ of $T-S$ containing $w$ and replace the leaf node of $D_v$
associated to $w$ with the root of $D_w$.
The result is a decision tree  $D'_v$ for the subtree of $T$ including $S_v$
and all the components of $T-S$ including some neighbor $w$ of $v$.

In order to obtain a decision tree $D_T$ for $T$ we now modify $D_S$ as follows:
for each vertex $v$ in $S$ with a neighbor not in $S$, replace the leaf in
$D_S$ associated with $v$ with the decision tree $D'_v$ computed above.

\begin{figure}
    \centering
    \includegraphics[width=\hsize]{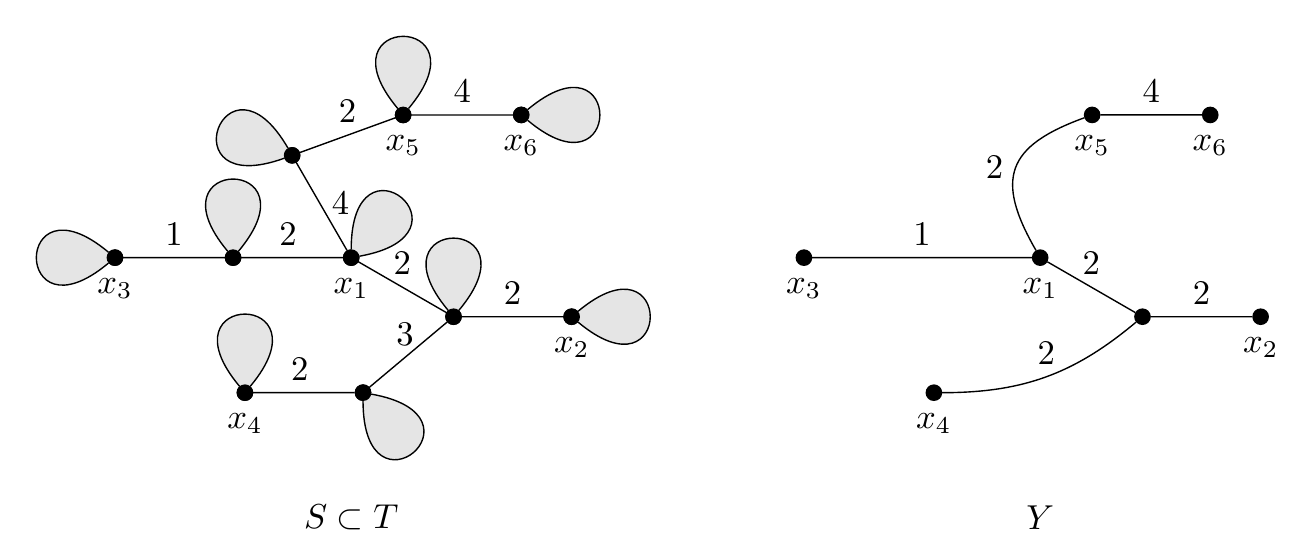}
   \caption{The tree $S$, the important set of vertices $X$ and the auxiliary tree $Y$}
   \label{fig:alg}
\end{figure}

\noindent
{\bf The Approximation guarantee for $D_T$.}
Let $\APP(T)=\frac{cost(D_T)}{OPT(T,c)}$ denote the approximation ratio obtained by Algorithm~\ref{alg} on the instance
$(T, c)$. Let $\APP(k)=\max_{|T|\le k} \APP(T)$.

\begin{lemma}\label{lem:apx}
For any tree $T$ on $n$ vertices and any cost assignment $\c$, we have $\APP(T) \leq 4 \log n/\log\log n$.
\end{lemma}

\begin{proof}
For every $1 \leq k \leq n$ let $f(k)= \max\{1, 4\log k/\log\log n\}$.
We shall prove by induction on $k$ that $APP(k)\le f(k)$, which implies the statement of the lemma.

If $|T|\le t$ then our algorithm builds an optimal decision tree,
thus $APP(k)=1\le f(k)$ for $k\le t$.
This establishes the induction base.

Choose a tree $T$ as in the statement of the lemma such that $APP(T) = APP(n)$.
Let $S$ and $Y$ be the substructures of $T$ built by the algorithm as described above.
Let $\tilde{V}$ be the set of vertices of $S$ with some neighbor not in $S$.
For each $w \not \in V(S)$ let $H_w$ be the maximal component of $T-S$ containing $w$.
Let ${\cal H}$ be the set of maximal components of $T-S$.
Then, by construction, we have
\begin{eqnarray}
\APP(T)&=&\frac{\ALG(T)}{\OPT(T)} \leq \frac{cost(D_S) + \max_{v \in \tilde{V}} cost(D_v) + \max_{w \not \in V(S)} cost(D_w)}{OPT(T,c)}  \label{eq:a}\\
&\leq& \frac{cost(D_S)}{OPT(S,c)} + \max_{v \in \tilde{V}} \frac{cost(D_v)}{OPT(S_v, c)} + \max_{w \not \in V(S)} \frac{cost(D_w)}{OPT(H_w, c)} \label{eq:b}\\
&\leq& 4+ \max_{H \in {\cal H}} \frac{ALG(H)}{OPT(H, c)} =  4+ \max_{H\in {\cal H}}\{\APP(H)\} \label{eq:c}\\
&\leq& 4 + \max_{H \in {\cal H}} f(|H|) \leq 4 +  f(|T|/\log n)  \label{eq:d} \\
&=&  4 + f(n/\log n) =  4 + \frac{4 \log \frac{n}{\log n}}{\log \log n} =  \frac{4 \log n}{\log \log n},  \label{eq:e}
\end{eqnarray}
where

\begin{itemize}
\item  (\ref{eq:b}) follows from (\ref{eq:a}) because of $OPT(S,c), OPT(S_v), OPT(D_w) \leq OPT(T,c)$ (Lemma \ref{lem:subtree})
\item  (\ref{eq:c}) follows from (\ref{eq:b}) because of (\ref{eq:S-approx}) and
the fact that any decision tree for a star $S_v$ has the same cost, hence also equal to $OPT(S_v, c)$
\item  in (\ref{eq:d}) the first inequality  follows by induction and the second inequality by Lemma~\ref{lem:H-components}
\item  (\ref{eq:e}) follows from (\ref{eq:d}) because of $|T| = n$ and the definition of $f(\cdot)$.
\end{itemize}
\qed\end{proof}

\begin{lemma}\label{lem:time}
For a tree $T$ on $n$ vertices, the Algorithm~\ref{alg} builds the decision tree $D_T$ in time polynomial in $n$.
\end{lemma}

The proof of Lemma~\ref{lem:time} is deferred to the appendix.
Lemma~\ref{lem:time} and Lemma~\ref{lem:apx} now imply Theorem~\ref{thm:main}.

\algnewcommand{\LineComment}[1]{\State \(\triangleright\) #1}
\begin{varalgorithm}{TS}
\caption{Tree Search Algorithm}\label{alg}
\begin{algorithmic}[1]
\small
\Function{Main}{tree $T$, cost $\c$}
\State $t \gets 2^{\lfloor \log \log |T| \rfloor +2}$
\State {\bf Output} $D \gets TreeSearch(T, \c, t)$
\EndFunction
\Function{TreeSearch}{tree $T$, costs $\c$, t}
\State {\bf	if} $|T|\le t$ {\bf then}  {\bf return} optimal decision tree $D_X$ for $T$ computed by Proposition~\ref{prop:exp}
	\State $S_0 \gets \emptyset$   \label{alg:sbuild1}
   \ForAll{$i=1,\dots,t$}
      \State $x_i \gets$ centroid of a maximum size connected component of $T-S_{i-1}$ 
      \State $S_i \gets$ smallest subtree containing $x_i$ and $S_{i-1}$
   \EndFor   \label{alg:sbuild2}
   \State $X \gets \{x_i|\; i=1,\dots,t\} \cup \{v\in V(S)|\; \deg_S(v)\ge 3\}$
   \State $Y \gets$ tree on vertex set $X$, $uv\in E(Y)$ iff $X\cap P_{u,v}=\{u,v\}$
   \ForAll {$uv\in E(Y)$}
      \State $\c_Y(uv) \gets \min_{e\in P_{u,v}} \c(e)$
      \State $e_{uv} \gets$ edge of $P_{u,v}$ with minimum cost
   \EndFor
   \State $Z \gets \bigcup_{uv\in E(Y)} e_{uv}$
   \State Compute optimal decision tree $D_Y$ for $(Y, \c_Y)$ by Proposition~\ref{prop:exp}
   \ForAll {$uv\in E(Y)$}
      \State Replace label of $uv$ in $D_Y$ by $e_{uv}$
   \EndFor
   \ForAll {$H$ connected component of $Y-Z$}
      \LineComment{$H$ contains at most one vertex of degree 3 or more, i.e., $H$ is a spider}
      \State Compute 2-approximate decision tree $D_H$ for $H$ by Lemma~\ref{lem:spider}   \label{alg:spider}
      \State replace the leaf $k\in D_Y$ corresponding to $H$ by the root of $D_H$
   \EndFor
      \ForAll {$v\in V(S)$ with a neighbor not in $S$}
      \State $S_v \gets$ star induced by $v$ and its neighbors outside of $V(S)$ \label{alg:sv}
			\State Construct decision tree $D_v$ for $(S_v,\c)$
      \ForAll {$w\in S_v\setminus \{v\}$}
         \State $U \gets $ connected component of $T-S$ containing $w$
         \State $D_w \gets TreeSearch(U, \c, t)$   \label{alg:recurse}
	 \State leaf of $D_v$ corresponding to $w$ $\gets$ root of $D_w$
      \EndFor
       \State replace the leaf of $D_Y$ associated to $v$ by the root of $D_v$
  \EndFor


   \Return $D_Y$
\EndFunction

\end{algorithmic}
\end{varalgorithm}
\normalsize

\section{Tree search with non-uniform costs is NP-hard on spider graphs}
In this section we provide a new hardness result which contributes to refining
the separation between hard and polynomial instances of the tree search problem
with non-uniform costs.
We show that the problem of finding a minimum cost decision tree is  hard even
for instances where the input graph is a spider and the length of every leg is three.

\smallskip

Our reduction is from the Knapsack Problem.
%
%
%
The input of the Knapsack Problem is given by:
a knapsack size $W$, a desired value $V$, and  a set of items, $(v_i,w_i)_{i \in [m]}$,
where $v_i$ is the value and $w_i$ is the weight of the $i$th item.
The goal is to decide whether there exists a subset of items of total value at least $V$ 
and whose weight can be contained in the knapsack,
i.e., whether there is a $J \subseteq [m]$ such that $\sum_{j\in J} w_j \le W$ and $\sum_{j\in J} v_j \ge V$.

From a knapsack instance we construct an instance $(S, \c)$ for the tree search problem with non-uniform costs,
where $S$ is a spider.
Each leg will correspond to an item.
Therefore, we will speak of the $i$th leg as the leg corresponding to the $i$th item.
For each $i \in [m]$, the $i$th leg will consist of three edges:
the one closest to the head will be called {\em femur} (and referred to as $f_i$),
the middle edge will be called {\em tibia} (and referred to as $t_i$),
the end will be called the {\em tarsus} (and referred to as $s_i$). 
The cost function is defined as follows:
For each $i \in [m]$, we set $c(f_i) = v_i + w_i$; $c(t_i) = v_i$ and $c(s_i) = N$, with
$N$ a large number to be determined later.

It is easy to see that in an optimal strategy,
for each $i \in [m]$ the edge $s_i$ is always queried last among the edges on the $i$th leg.
Given a decision tree $D$,
we denote by $I^D$ the set of indices of the legs for which,
in $D$, the node associated with the query to the tibia is an ancestor of the node associated with the
query to the femur.
Then,
we have the following proposition,
whose proof is deferred to the appendix.
%
\begin{proposition} \label{prop:knap-dectreestructure}
There is an optimal decision tree $D$ with $I^D \neq \emptyset$ and such that:

(i) for any $i \in I^D$ and $j \in [m] \setminus I^D$ the node of D associated
with the $j$th femur is an ancestor of the node associated with the $i$th tibia.

(ii) for any $i, j \in I^D$ the node of D associated with the $i$th tibia is an
ancestor of the node associated with the $j$th femur.

\end{proposition}

By this proposition, we can assume that in the optimal decision tree $D$
for at least one leg of the spider the first edge queried is a tibia. In addition, in $D$,
there is a root to leaf path where first all femurs not in $I^D$ are queried,
then all tibias in $I^D$ and finally all femurs in $I^D$
(see Fig.~\ref{fig:knap-spider} in Appendix for a pictorial example).
Then, the cost of such a decision tree is given by the maximum between the cost of the leaf on the legs with index in $I^D$
and whose tibia is queried as last,
and the cost of the central vertex of the spider.
It follows that the cost of the optimal solution is given by the following expression
$$
OPT(S, \c) = \min_{\emptyset \subset I \subseteq [m]} \max \left\{N + \sum_{i \not \in I} (v_i+w_i) + \sum_{i \in I} v_i ;
\sum_{i \in I} v_i + \sum_{i \in [m]} (v_i + w_i) \right\}
$$

If we set $N = \sum_{i \in [m]} (v_i + w_i) - W - V$, then we can rewrite the above expression as follows:
$$
OPT(S, \c) = \min_{\emptyset \subset I \subseteq [m]} \max \left\{N + \sum_{i \not \in I} w_i + \sum_{i \in [m]} v_i ;
N + W + V  + \sum_{i \in I} v_i  \right\}
$$

Now, it is easy to see that $OPT(S, \c)$  is at most $\sum_{i \in [m]} v_i+N+W$ if an only if
$\sum_{i \not \in I} w_i \leq W$ and $\sum_{i \not \in I} v_i \geq V$, that is, if and only if
the set $[m] \setminus I$ is a solution for the knapsack problem.
Note that as the values and weights are unrelated, we can indeed choose $N$ as big as necessary for the above reduction,
which is clearly polynomial in the size of the input to the knapsack problem.

%
%

\section*{Acknowledgment}
We are very grateful to Bal\'azs Patk\'os for organizing $5^{\mathrm{th}}$  Eml\'ekt\'abla Workshop 
where we collaborated on this paper.

%

\begin{thebibliography}{1}

\bibitem{CicaleseJLV12}
Ferdinando Cicalese, Tobias Jacobs, Eduardo~Sany Laber, and Caio~Dias Valentim.
\newblock The binary identification problem for weighted trees.
\newblock {\em Theor. Comput. Sci.}, 459:100--112, 2012.

\bibitem{laTGreSch95}
Pilar de~la Torre, Raymond Greenlaw, and Alejandro Sch{\"a}ffer.
\newblock Optimal edge ranking of trees in polynomial time.
\newblock {\em Algorithmica}, 13(6):592--618, 1995.

\bibitem{Dereniowski06}
Dariusz Dereniowski.
\newblock Edge ranking of weighted trees.
\newblock {\em Discrete Appl. Math.}, 154:1198--1209, May 2006.

\bibitem{Dereniowski08}
Dariusz Dereniowski.
\newblock Edge ranking and searching in partial orders.
\newblock {\em Discrete Applied Mathematics}, 156(13):2493--2500, 2008.

\bibitem{IyeRatVij91}
Ananth~V. Iyer, H.~Donald Ratliff, and Gopalakrishnan Vijayan.
\newblock On an edge ranking problem of trees and graphs.
\newblock {\em Discrete Applied Mathematics}, 30(1):43--52, 1991.

\bibitem{Lam98optimaledge}
Tak~Wah Lam and Fung~Ling Yue.
\newblock Optimal edge ranking of trees in linear time.
\newblock In {\em SODA '98: Proceedings of the ninth annual ACM-SIAM symposium
  on Discrete algorithms}, pages 436--445, Philadelphia, PA, USA, 1998. Society
  for Industrial and Applied Mathematics.

\bibitem{Makino2001}
Kazuhisa Makino, Yushi Uno, and Toshihide Ibaraki.
\newblock On minimum edge ranking spanning trees.
\newblock {\em J. Algorithms}, 38:411--437, February 2001.

\end{thebibliography}


\begin{thebibliography}{1}


\bibitem{Ahlswede} R. Ahlswede, I. Wegener.
{\em Search Problems}. J. Wiley \& Sons, Chichester--New York, 1987.

\bibitem{Aigner}
M. Aigner. {\em Combinatorial Search}. Wiley--Teubner,
 New York--Stuttgart, 1988.

\bibitem{Ben-Asher}
Y. Ben-Asher, E. Farchi, and I. Newman.
Optimal search in trees.
{\em SIAM Journal on Computing}, 28(6):2090--2102, 1999.

\bibitem{CicaleseJLV12}
F. Cicalese, T. Jacobs, E. Laber, and C. Valentim.
\newblock The binary identification problem for weighted trees.
\newblock {\em Theoretical Computer Science}, 459:100--112, 2012.

\bibitem{laTGreSch95}
P. de~la Torre, R. Greenlaw, and A. Sch{\"a}ffer.
\newblock Optimal edge ranking of trees in polynomial time.
\newblock {\em Algorithmica}, 13(6):592--618, 1995.

\bibitem{Dereniowski06}
D. Dereniowski.
\newblock Edge ranking of weighted trees.
\newblock {\em Discrete Applied Mathematics}, 154:1198--1209, May 2006.

\bibitem{Dereniowski08}
D. Dereniowski.
\newblock Edge ranking and searching in partial orders.
\newblock {\em Discrete Applied Mathematics}, 156(13):2493--2500, 2008.

\bibitem{Faigle}
U. Faigle, L. Lov\'asz, R. Schrader, Gy. Tur\'an.
Searching in trees, series-parallel and interval orders.
{\em SIAM Journal on Computing}, 15(4):1075--1084, 1986.

\bibitem{IyeRatVij91}
A.V. Iyer, H.D. Ratliff, and G. Vijayan.
\newblock On an edge ranking problem of trees and graphs.
\newblock {\em Discrete Applied Mathematics}, 30(1):43--52, 1991.


\bibitem{Knuth} D. Knuth. {\em Searching and Sorting},
3rd volume of The Art of Computer Programming,
Addison-Wesley, 1998.

\bibitem{Lam98optimaledge}
T. Wah Lam and F.~Ling Yue.
\newblock Optimal edge ranking of trees in linear time.
\newblock In {\em SODA '98: Proceedings of the ninth annual ACM-SIAM symposium
  on Discrete algorithms}, pages 436--445, Philadelphia, PA, USA, 1998. Society
  for Industrial and Applied Mathematics.

\bibitem{Linial} N. Linial and M. Saks.
Searching order structures.
{\em Journal of Algorithms}, 6:86--103, 1985.

\bibitem{Makino2001}
K. Makino, Y. Uno, and T. Ibaraki.
\newblock On minimum edge ranking spanning trees.
\newblock {\em Journal of Algorithms}, 38:411--437, February 2001.

\bibitem{Mozes}
S. Mozes, K. Onak, O. Weimann.
Finding an optimal tree searching strategy in linear time.
In {\em Proc. of the 19th Annual ACM-SIAM Symp. om Discrete Algorithms} (SODA'08),
pp. 1096--1105, 2009.

\bibitem{Wermelinger}
M. Wermelinger.
Searching Efficiently in Posets.
{\em Topics in Programming Techonology}, New University of Lisbon, 1993.

\end{thebibliography}

\newpage
\appendix

\section*{Appendix}

\subsection*{The Proof of Proposition \ref{prop:contract}}

\begin{proof}
Let $D$ be an optimal decision tree for the instance $(T, \c)$.
Let us assume without loss of generality that in $D$ the node $\nu_1$
associated with $e_1=\{u_1, v\}$ is an ancestor of the node $\nu_2$ associated
with $e_2 = \{u_2, v\}$.
Notice that one of the children of $\nu_2$ is a leaf associated with the vertex $v$.
Let $\tilde{D}$ be the subtree of
$D$ rooted at the non-leaf child of $\nu_2$.

Let $D'$ be the decision tree obtained from $D$ by associating the node $\nu_1$ to the edge $e = \{u_1, u_2\}$ and replacing the
subtree rooted at $\nu_2$ with the subtree $\tilde{D}$.

It is not hard to see that $D'$ is a proper decision tree for $T'$.
In addition we also have that for any vertex $z$ of $T'$ which is associated to
a leaf in $\tilde{D}$ it holds that $cost^{D'}(z) = cost^{D}(z) - \c(e_1) - \c(e_2) +  \c'(u_1u_2)$,
and for any other vertex $z$ of $T'$ we have
$cost^{D'}(z) = cost^{D}(z) - \c(e_1) + \c'(u_1u_2)$ or $cost^{D'}(z) = cost^{D}(z)$.
It follows that $OPT(T', \c') \leq cost(D') \leq  cost(D) = OPT(T, \c)$.
\qed\end{proof}

\subsection*{The Proof of Lemma \ref{lem:time}}

\begin{proof}
If $|T|\le t$ then the algorithm builds an optimal  decision tree for $T$ 
in time $O(2^{t}\cdot t) = O(n^4)$ using the construction from Proposition~\ref{prop:exp}.
Otherwise, every iteration needed to build the subtree $S$ (lines \ref{alg:sbuild1}--\ref{alg:sbuild2} of the algorithm)
introduces one new vertex $x_i$ and at most one other vertex of degree at least three, thus $|X| \leq 2t+1$.
Proposition~\ref{prop:exp} then implies that an optimal decision tree $D_Y$ for $Y$
can be computed in time $O(2^{2t}\cdot 2t)$ which is polynomial in $n$.
By Lemma~\ref{lem:spider}, the $2$-approximation decision tree $D_H$ for $H$ can be computed in $O(n^2)$ time.
Building the decision tree $D_v$ for the stars $S_v$  takes $O(|S_v|)$ time (line~\ref{alg:sv}).
The rest of the algorithm, not counting the recursion on line~\ref{alg:recurse},
needs time $O(n^2)$.
As the recursion is for a graph whose size is at most half of the original,
the overall algorithm running time is polynomial in $n$.
\qed\end{proof}

\subsection*{The Proof of Proposition \ref{prop:knap-dectreestructure}}

\begin{figure}
    \centering
    \includegraphics{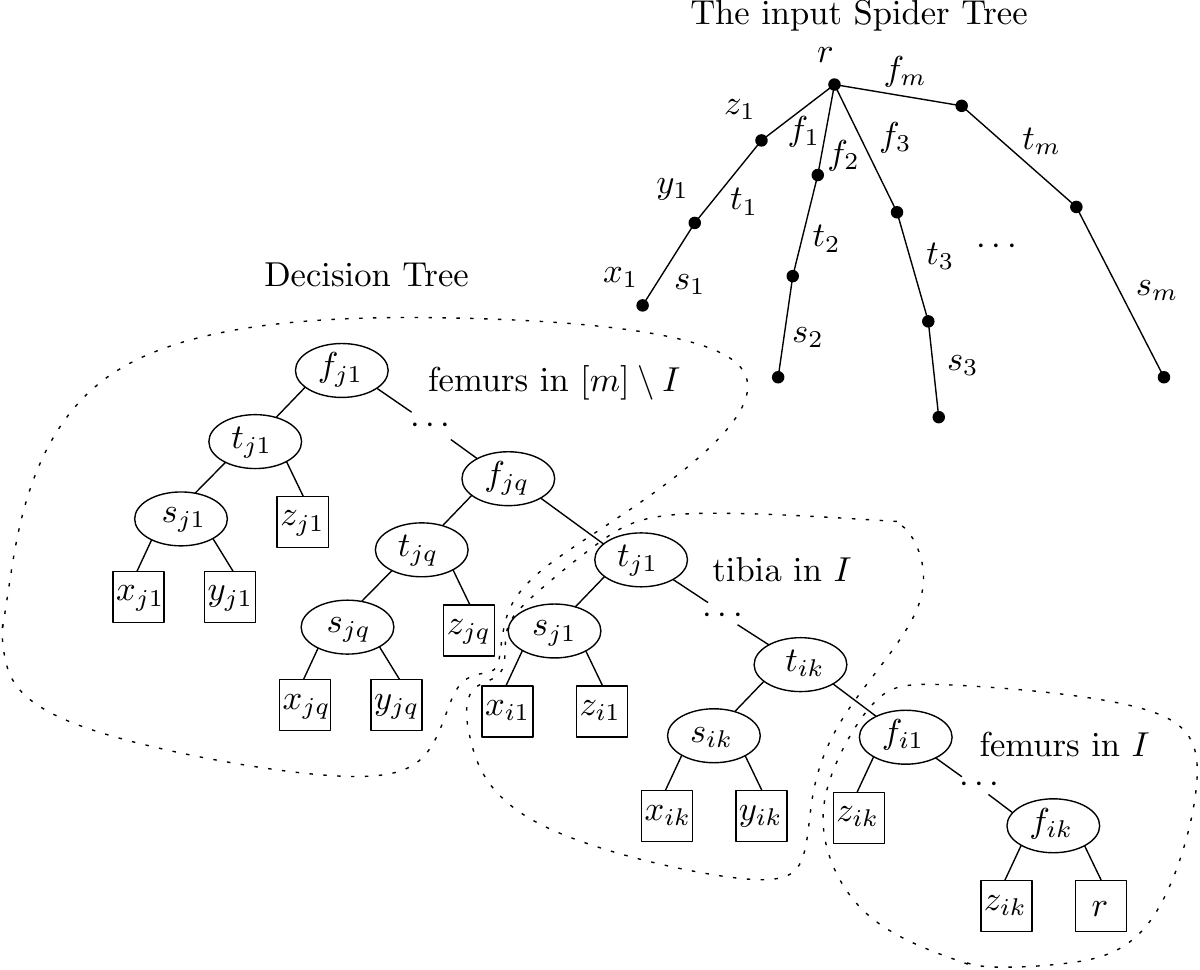}
   \caption{The structure of the optimal decision tree in Proposition \ref{prop:knap-dectreestructure}.
   For the ease of notation, we use $I$ for $I^D$.
   The cost of this decision tree can be obtained as the max of the costs provided by the leaf associated to $x_{i_k}$ and the
   leaf associated with $r$.}
   \label{fig:knap-spider}
\end{figure}




\begin{proof}
We first show that there is an optimal decision tree with $I^D \ne \emptyset$. 
Let $D^*$ be a decision tree where each femur is queried before the corresponding tibia, i.e., $I^{D^*} = \emptyset$.
Let $i$ be the index of the last femur queried. Therefore one of the two children of the node querying $f_i$ is a
leaf associated to $r$, while in the subtree rooted at the other child the leaves are associated to the vertices in the $i$th leg.
Let $z_i, y_i, x_i$, denote the vertices on the $i$th leg in order of increasing distance from $r$.
It is not hard to see that
$$\max_{v \in \{z_i, y_i, x_i, r\}} cost^{D^*}(v)=  K+c(f_i)+c(t_i)+c(s_i),$$ where $K$ is the
cost of the queries on the path from the root of $D^*$ to the parent of the node associated with the query to $f_i$.

Now consider the decision tree obtained from $D^*$ by replacing the query to $f_i$ with a query to $t_i$,
then one child of this node queries
$f_i$ and the other child queries $s_i$. Let $D'$ be the resulting decision tree. It is not difficult to see that we now have
\begin{eqnarray*}
\max_{v \in \{z_i, y_i, x_i, r\}} cost^{D'}(v) &=& \max\{K+ c(t_i)+c(s_i), K+ c(t_i)+c(f_i)\} \\
&\leq& \max_{v \in \{z_i, y_i, x_i, r\}} cost^{D^*}(v)
\end{eqnarray*}
and $cost^{D'}(v) = cost^{D^*}(v)$ for any $v \not \in \{z_i, y_i, x_i, r\}$. Hence $cost(D') \leq cost(D^*)$ with $I^{D'} \neq \emptyset$ for $D'$.

Now, assuming that $I = I^D \neq \emptyset$, we can show (i) and (ii). First we observe that if at least one of (i) and (ii) does not hold then at least one of the
following conditions holds:
\begin{itemize} 
\item[(i')] there exists $i \in I$ and $j \in [m]\setminus I$ such that the node $\nu_j$ associated with $f_j$ is a child of the node $\nu_i$ associated with $t_i$;
\item[(ii')] there exists $i,j \in I$ such that the node $\nu_i$ associated with $t_i$ is a child of the node $\nu_j$ associated with $f_j$;
\item[(iii')] there exists $i \in I$ and $j \in [m]\setminus I$ such that the node $\nu_j$ associated with $f_j$ is a child of the node $\nu_i$ associated with $f_i$.
\end{itemize}
Indeed, if none of these three conditions holds then  (i) and (ii) follow.

Therefore, it is enough to show that if we have an optimal tree where one of the three conditions holds, by swapping the nodes $\nu_i$ and $\nu_j$
involved, we can obtain a new decision tree whose total cost is not larger than the cost of the original decision tree.
This implies that by repeated use of this
swapping procedure, we have an optimal  decision tree where both (i) and (ii) hold.

We shall limit to explicitly show this argument for the case where in the optimal decision tree $D^*$ condition (i') holds.
Therefore, we have
$$
\max_{v \in \{z_j, y_j, x_j\}} cost^{D^*}(v) = K+c(t_i) + c(f_j) + c(t_j) + c(s_j)
$$
$$
cost^{D^*}(x_i) =  cost^{D^*}(y_i) = K + c(t_i)+c(s_i)
$$

Let $D'$ be the decision tree obtained after swapping the queries to $f_j$ and
the query to $s_i$ so that now the latter is the parent of the former.
Therefore, we have
$$
\max_{v \in \{z_j, y_j, x_j\}} cost^{D'}(v) = K+ c(f_j) + c(t_j) + c(s_j)
$$
$$
cost^{D'}(x_i) =  cost^{D'}(y_i) = K + c(f_j)+ c(t_i)+c(s_i)
$$

and for each $v \not \in \{z_j, y_j, x_j, y_i, x_i\}$ it holds that $cost^{D^*}(v) = cost^{D'}(v)$.

Since $c(s_i) = c(s_j)$ we have that
$$\max_{v \in \{z_j, y_j, x_j, y_i, x_i\}} cost^{D'}(v) \leq \max_{v \in \{z_j, y_j, x_j, y_i, x_i\}} cost^{D^*}(v),$$ hence
$cost(D') \leq cost(D^*)$.

We can use an analogous argument to show that we can swap queries in order to have an optimal decision tree where neither (ii') nor (iii') holds.
The resulting tree satisfies (i) and (ii) as desired.
\qed\end{proof}

\end{document}